\documentclass[%
superscriptaddress,
preprint,
 amsmath,amssymb,
 aps,
prb,
]{revtex4-2}

\usepackage{graphicx}
\usepackage{dcolumn}
\usepackage{bm}
\usepackage{braket}
\usepackage{comment} 
\usepackage{hyperref}
\usepackage{siunitx}

\usepackage[shortlabels]{enumitem}

\usepackage{booktabs}
\newcolumntype{P}[1]{>{\centering\arraybackslash}p{#1}}

\usepackage{xcolor}

\usepackage{nicematrix}
\NiceMatrixOptions{cell-space-limits = 1pt}

\usepackage{mathtools}

\usepackage{amsthm}

\newtheorem*{theorem*}{Theorem}

\theoremstyle{definition}

\theoremstyle{remark}

\newtheorem*{example*}{Example}
\newtheorem*{algorithm*}{Algorithms}

\usepackage{etoolbox}

\makeatletter
\patchcmd{\frontmatter@abstract@produce}
  {\vskip200\p@\@plus1fil
   \penalty-200\relax
   \vskip-200\p@\@plus-1fil}
  {}
  {}
  {}
\makeatother

\begin{document}


\title{Passivity constraints on the relations between \\ transmission, reflection, and absorption eigenvalues}

\author{Cheng Guo}
\email{guocheng@stanford.edu}
\affiliation{Ginzton Laboratory and Department of Electrical Engineering, Stanford University, Stanford, California 94305, USA}

\author{Shanhui Fan}
\email{shanhui@stanford.edu}
\affiliation{Ginzton Laboratory and Department of Electrical Engineering, Stanford University, Stanford, California 94305, USA}%

\date{\today}

\begin{abstract}
We investigate the passivity constraints on the relations between transmission, reflection, and absorption eigenvalues in linear time-invariant systems. Using techniques from matrix analysis, we derive necessary and sufficient conditions for the permissible combinations of these eigenvalues. Our analysis reveals that the set of allowable eigenvalue combinations forms a convex polyhedron in eigenvalue space, characterized by a trace equality and a set of linear inequalities. Surprisingly, we uncover a direct connection between this physical problem and Alfred Horn's inequalities, a fundamental result in matrix theory. We provide explicit examples for systems with varying numbers of input ports, demonstrating the increasing complexity of the constraints as system size grows. We apply our theory to analyze the implications of important phenomena, including open and closed channels, coherent perfect reflection and reflectionless scattering modes, and coherent perfect absorption and coherent zero absorption. Our findings not only offer a complete characterization of passivity constraints on wave transport eigenvalues but also establish an unexpected bridge between fundamental wave physics and advanced matrix theory, opening new avenues for research at their intersection. These results have significant implications for the design and optimization of passive wave devices across a wide range of applications in optics, acoustics, and mesoscopic physics.
\end{abstract}

\maketitle


\section{Introduction}\label{sec:introduction}

Passivity is a fundamental property of many physical systems, characterizing their inability to generate energy~\cite{foster1924,meixner1954,wohlers1965,willems1972,willems1972a,youla1959a,boyd1982,zemanian1987,zemanian1996,nedic2020,srivastava2021}. This intuitive concept imposes significant constraints on system behavior~\cite{bernland2011}, with far-reaching implications in physics and engineering~\cite{chao2022}. For linear time-invariant systems, passivity implies causality~\cite{guttinger1966,wohlers1969} and places restrictions on the system's transfer function~\cite{wohlers1969}. These restrictions can be related to Herglotz functions~\cite{aronszajn1956,zemanian1987, gesztesy2000, nedic2020}, which are analytic functions mapping the upper half-plane into itself~\cite{titchmarsh1948,kac1974r}. The integral identities of Herglotz functions yield physical constraints in the form of dispersion relations~\cite{kronig1926,kramers1927diffusion,gorter1936,toll1956,sharnoff1964,guttinger1966,hu1989,king2006,king2009a,king2009b,nussenzveig2012} and sum rules~\cite{gustafsson2010}. 

Understanding these passivity constraints is crucial for both theoretical and practical applications~\cite{chao2022}. In theoretical physics, passivity constraints provide insights into fundamental phenomena such as constitutive relations~\cite{altarelli1972,landau2008,shim2021}, extinction~\cite{sohl2007,shim2019,kuang2020}, scattering~\cite{sohl2007,bernland2011a,miller2016,kuang2020,kuang2020a,molesky2020,zhang2023g}, absorption~\cite{purcell1969,miller2016,ivanenko2019}, local density of states~\cite{scheel2008,miller2016}, spontaneous emission~\cite{barnett1996,scheel2008}, radiative heat transfer~\cite{molesky2020a,zhang2023g}, and the speed-of-light limitation~\cite{glasgow2001,yaghjian2007,welters2014}. In engineering applications, these constraints guide the design of various devices and systems, including matching networks~\cite{fano1950,guillemin1957}, radar absorbers~\cite{purcell1969,rozanov2000,nordebo2014}, filters~\cite{youla2016}, multiplexers~\cite{youla2016}, passive optical networks~\cite{lam2007,kramer2005,li2017h,faruk2021}, high-impedance surfaces~\cite{brewitt-taylor2007}, passive metamaterials~\cite{gustafsson2010}, antennas~\cite{gustafsson2007,doane2013,jonsson2013}, cloaks~\cite{fleury2014a,fleury2015,monticone2016,cassier2017a,jelinek2021,strekha2024}, and periodic structures~\cite{gustafsson2012}.

In the context of wave transport, transmission, reflection, and absorption are fundamental processes~\cite{datta1995,bohren2004,chen2005,carminati2021}. These phenomena are characterized by their respective eigenvalues, which determine a wide variety of transport properties~\cite{beenakker1997,rotter2017}. For an $n$-port linear time-invariant system described by field transmission and reflection matrices $t$ and $r$, the transmission, reflection, and absorption eigenvalues are defined as the eigenvalues of $T = t^{\dagger}t$, $R = r^{\dagger}r$, and $A = I - t^{\dagger}t - r^{\dagger}r$, respectively~\cite{beenakker1997,beenakker1998,rotter2017,cao2022}.

These transport eigenvalues play a central role in various physical phenomena and applications. In mesoscopic electron transport~\cite{berkovits1994,vanrossum1999,datta1995}, the transmission process can be viewed as parallel transmission over orthogonal eigenchannels~\cite{landauer1970,datta2013}. The universal bimodal distribution of transmission eigenvalues in lossless diffusive systems~\cite{dorokhov1984,pichard1986,muttalib1987,mello1989,nazarov1994} leads to celebrated phenomena such as open and closed channels~\cite{dorokhov1984,imry1986,deboer1994,vellekoop2008}, universal conductance fluctuations~\cite{lee1985b,imry1986}, and quantum shot noise~\cite{beenakker1992,altshuler1994}. The disorder dependence of these eigenvalues underpins fundamental phenomena like weak localization~\cite{barabanenkov1973,kuga1984,albada1985,wolf1985} and Anderson localization~\cite{anderson1958,abrahams1979,altland2005}. 

In classical wave systems, absorption significantly affects the distribution of transmission~\cite{brouwer1998} and reflection eigenvalues~\cite{bruce1996,beenakker1996,goetschy2013}. Recent advances in wave-front shaping techniques~\cite{vellekoop2007,yu2017e,reck1994,morizur2010} have enabled experimental access to transport eigenchannels, leading to the realization of novel phenomena such as coherent enhancement or suppression of transmission~\cite{vellekoop2008,popoff2010,aulbach2011,kim2012a,shi2012,yu2013a,gerardin2014,pena2014a,popoff2014,davy2015,shi2015a,hsu2015,liew2016,mounaix2016,bender2020c}, reflection~\cite{horodynski2022}, and absorption~\cite{chong2011a}, as well as coherent perfect absorption~\cite{chong2010a,wan2011,sun2014,baranov2017,mullers2018,pichler2019,sweeney2019,chen2020f,wang2021,slobodkin2022}, coherent perfect reflection~\cite{yan2014a,dai2018,guo2024b}, and reflectionless scattering modes~\cite{sweeney2020a,stone2021}. Furthermore, these transport eigenvalues are crucial in recent developments in photonic control of thermal radiation~\cite{beenakker1998,miller2017,fan2017,liu2022a,guo2022b,guo2023b}, unitary control of multiport coherent~\cite{guo2022b,guo2024d} and partially coherent waves~\cite{guo2024a,guo2024b}, reciprocity constraints on reflection~\cite{guo2022a}, and topology of scattering matrices~\cite{guo2023c,guo2024c}. They have found applications in diverse fields such as imaging~\cite{sebbah2001a,ntziachristos2010a,cizmar2012,kang2015,guo2018,guo2018a,yoon2020b,wang2020p,long2021,bertolotti2022,wang2022,long2022}, communication~\cite{miller2013c,miller2019,seyedinnavadeh2024}, sensing~\cite{aulbach2011,sarma2015,mounaix2016,jeong2018a,liu2020s}, energy harvesting, integrated photonics~\cite{miller2013b,carolan2015,turpin2018,leedumrongwatthanakun2020,wetzstein2020a}, and optical neural networks~\cite{shen2017,prabhu2020,zhang2021f,ashtiani2022,bandyopadhyay2022,ohno2022,chen2023,pai2023}.

Despite significant efforts in studying both passivity and transport eigenvalues, the question of passivity constraints on transport eigenvalues remains open. Specifically, a fundamental question persists: what are the permissible combinations of transmission, reflection, and absorption eigenvalues in passive linear time-invariant systems?

This question presents two significant challenges. First, these three sets of eigenvalues are not independent but intricately related~\cite{rammal1987,freilikher1994,bruce1996,paasschens1996,brouwer1998,kogan2000,liew2014,liew2014a,hsu2015,liew2015a,yamilov2016,rotter2017}. A comprehensive understanding of their interrelations has remained elusive. Second, conventional methods for determining passivity constraints based on complex analysis of Herglotz functions~\cite{bernland2011} are not directly applicable to this problem. These \emph{analytic} methods cannot directly reveal the passivity constraints on transport eigenvalues, key \emph{algebraic} properties of the system. 

In this paper, we provide a complete answer to this question. We address the challenges by adopting a fundamentally different approach to deriving passivity constraints based on linear algebra rather than complex analysis. Using matrix analysis techniques, we derive the complete set of passivity constraints on the three vectors of eigenvalues. We uncover a surprising connection between this physical problem and Alfred Horn's inequalities~\cite{horn1962,bhatia2001}, a fundamental result in matrix theory that has only been fully proven in recent decades~\cite{klyachko1998,knutson1999,fulton2000}. This connection opens up new avenues for research at the intersection of these fields and provides powerful tools for analyzing and designing passive wave systems.

The rest of the paper is organized as follows: Section~\ref{sec:physical_problem} provides a detailed description of the physical problem. Section~\ref{sec:math_form} presents its mathematical formulation. Section~\ref{sec:examples} provides examples and illustrations of the results for systems of different sizes. Section~\ref{sec:main_results} presents our main results, including the derivation of the passivity constraints on transmission, reflection, and absorption eigenvalues. Section~\ref{sec:discussions_Horn_ineq} explores the connection to Horn's inequalities and discusses their implications. Section~\ref{sec:applications} demonstrates the application of our theory to specific wave transport phenomena. Finally, Section~\ref{sec:conclusion} discusses the broader implications of our findings and potential future directions.

\section{Physical Problem}\label{sec:physical_problem}

\begin{figure}[htbp]
    \centering
    \includegraphics[width=0.5\textwidth]{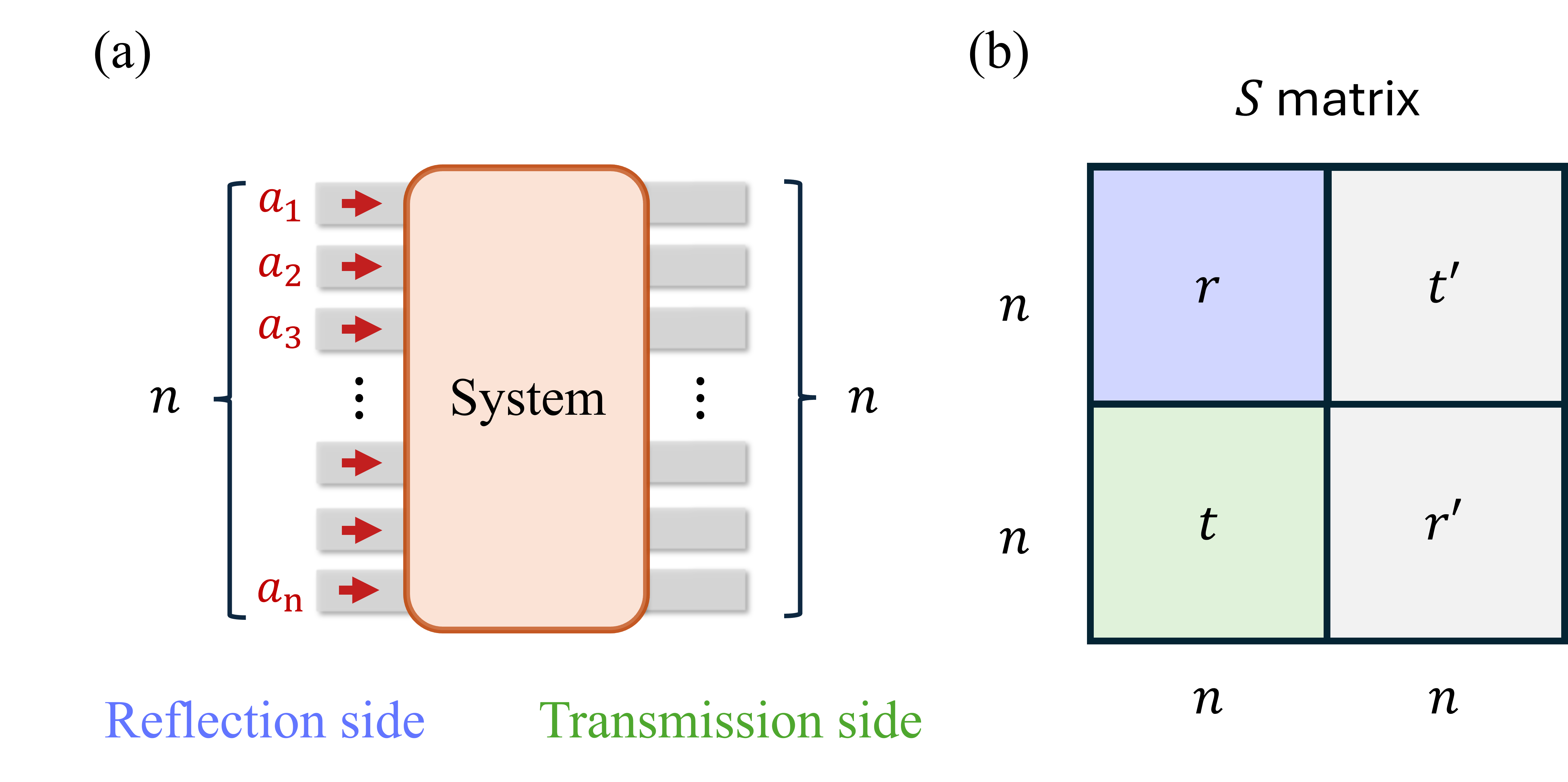}
    \caption{System and scattering matrix. (a) A $2n$-port passive linear time-invariant system with $n$ ports on each side. Waves are incident from the left. (b) Structure of the scattering matrix $S$. $r$ and $t$ ($r'$ and $t'$) represent reflection and transmission matrices for left (right) incidence, respectively. }
    \label{fig:system}
\end{figure}

Consider a $2n$-port passive linear time-invariant system with $n$ ports on each side, as illustrated in Fig.~\ref{fig:system}(a). The system is described by a $2n \times 2n$ scattering matrix
\begin{equation}
S = \begin{pmatrix}
r & t' \\
t & r' 
\end{pmatrix},    
\end{equation}
where $r$ and $t$ are the $n\times n$ field reflection and transmission matrices for waves incident from the left, while $r'$ and $t'$ are for waves incident from the right. In this paper, we focus on waves input from the left side, referring to this configuration as an $n$-input-port system. This distinguishes it from the general case where waves can enter from any of the $2n$ ports.

The transmission eigenvalues are defined as the eigenvalues of the transmittance matrix:
\begin{equation}\label{eq:def_transmittance_matrix}
T = t^{^{\dagger}} t.    
\end{equation}
The reflection eigenvalues are the eigenvalues of the reflectance matrix:
\begin{equation}\label{eq:def_reflectance_matrix}
R = r^{\dagger}r,    
\end{equation}
The absorption eigenvalues are the eigenvalues of the absorptivity matrix:
\begin{equation}\label{eq:def_absorptivity_matrix}
A = I - t^{^{\dagger}} t - r^{\dagger}r. 
\end{equation}
Note that all these matrices are Hermitian, hence their eigenvalues are real. We denote these three vectors of eigenvalues as $\bm{\tau}$, $\bm{\rho}$, and $\bm{\alpha}$, with components arranged in nonincreasing order:
\begin{equation}
\bm{\tau} = \bm{\lambda}(T), \quad \bm{\rho} = \bm{\lambda}(R), \quad \bm{\alpha} = \bm{\lambda}(A).    
\end{equation}
The tuple $(\bm{\tau}, \bm{\rho}, \bm{\alpha})$ characterizes the transmission, reflection, and absorption properties of the system for waves incident from the left side.

The question is: What are all possible values of the tuple $(\bm{\tau}, \bm{\rho}, \bm{\alpha})$ for all passive linear time-invariant systems?

\section{Mathematical formulation}\label{sec:math_form}

We now reformulate the physical problem into a mathematical form. Our focus is on the eigenvalues of the matrices $T$, $R$, and $A$ in passive linear time-invariant systems. By their definitions in Eqs.~(\ref{eq:def_transmittance_matrix}) and (\ref{eq:def_reflectance_matrix}), $T$ and $R$ are positive semidefinite, denoted as~\cite{bhatia2007}:  
\begin{equation}~\label{eq:condition-psd-T-R}
T \ge O, \quad R\ge O.   
\end{equation}
Passivity requires that $A$ is also positive semidefinite~\cite{raisbeck1954,youla1959a,beltrami1967,wohlers1969,boyd1982}:
\begin{equation}~\label{eq:condition-psd}
\quad A\ge O.    
\end{equation}
Energy conservation requires the sum of the three matrices to be the identity matrix:
\begin{equation}~\label{eq:condition-id}
T + R + A = I.    
\end{equation}
Our objective is to determine the set of all attainable tuples:
\begin{equation}
\Omega_{n} \coloneqq  \{ (\bm{\tau}, \bm{\rho}, \bm{\alpha} ) \} = \{ (\bm{\lambda}(T),\bm{\lambda}(R),\bm{\lambda}(A)) \}
\end{equation}
subject to conditions (\ref{eq:condition-psd-T-R}), (\ref{eq:condition-psd}) and (\ref{eq:condition-id}).

Conditions (\ref{eq:condition-psd-T-R}), (\ref{eq:condition-psd}) and (\ref{eq:condition-id}) immediately imply that all $\tau_i$, $\rho_i$, and $\alpha_i$ lie between 0 and 1:
\begin{equation}
\bm{\tau}, \bm{\rho}, \bm{\alpha} \in \mathcal{D}_{n},     
\end{equation}
where 
\begin{equation}
\mathcal{D}_{n} = \{(x_{1},\dots,x_{n})\in \mathbb{R}^{n}: 1 \geq x_{1}\geq \dots \geq x_{n} \geq 0\}.   
\end{equation}
Therefore,
\begin{equation}
\Omega_{n} \subseteq \mathcal{D}_{n} \times \mathcal{D}_{n} \times \mathcal{D}_{n}.     
\end{equation}

Additional constraints from (\ref{eq:condition-psd-T-R}), (\ref{eq:condition-psd}) and (\ref{eq:condition-id}) further restrict $\Omega_{n}$. The simplest of these is the trace equality:
\begin{equation}\label{eq:trace_equality}
\sum_{i=1}^{n} (\tau_{i} + \rho_{i} + \alpha_{i}) = n,    
\end{equation}
However, numerous other inequality constraints exist. Determining $\Omega_{n}$ necessitates a comprehensive understanding of all constraints implied by (\ref{eq:condition-psd-T-R}), (\ref{eq:condition-psd}) and (\ref{eq:condition-id}).

We employ a step-by-step strategy to determine $\Omega_{n}$:
\begin{enumerate}
    \item Keep $\bm{\tau}$ and $\bm{\rho}$ free and determine the set of attainable $\bm{\alpha}$, denoted as $\{\bm{\alpha}\}$.
    \item Fix $\tilde{\bm{\alpha}} \in \{\bm{\alpha}\}$, keep $\bm{\tau}$ free, and determine the set of attainable $\bm{\rho}$, denoted as $\{\bm{\rho}|\tilde{\bm{\alpha}}\}$.
    \item Fix $\tilde{\bm{\alpha}} \in \{\bm{\alpha}\}$ and $\tilde{\bm{\rho}} \in \{\bm{\rho}|\tilde{\bm{\alpha}}\}$, and determine the set of attainable $\bm{\tau}$, denoted as $\{\bm{\tau}|\tilde{\bm{\rho}},\tilde{\bm{\alpha}}\}$.
\end{enumerate}
This approach yields:
\begin{equation}
\Omega_{n} = \{(\tilde{\bm{\tau}},\tilde{\bm{\rho}},\tilde{\bm{\alpha}}) \in \mathcal{D}_{n} \times \mathcal{D}_{n} \times \mathcal{D}_{n}: \tilde{\bm{\alpha}} \in \{\bm{\alpha}\}, \tilde{\bm{\rho}} \in \{\bm{\rho}|\tilde{\bm{\alpha}}\}, \tilde{\bm{\tau}} \in \{\bm{\tau}|\tilde{\bm{\rho}}, \tilde{\bm{\alpha}}\}\}.     
\end{equation}

We make several remarks regarding this strategy.

First, our chosen sequence for determining the ranges of $\bm{\alpha}$, $\bm{\rho}$, and $\bm{\tau}$ is arbitrary. The result remains invariant under different sequences due to the problem's symmetry under the permutation of $T$, $R$, and $A$. Thus, if $(\bm{x},\bm{y},\bm{z}) \in \Omega_{n}$, any permutation of $(\bm{x},\bm{y},\bm{z})$ also belongs to $\Omega_{n}$. However, addressing $\bm{\alpha}$ first is physically useful in illustrating absorption's effect on reflection and transmission behaviors.

Second, while this step-by-step strategy for determining tuple ranges is general, its effectiveness depends on the specific problem. Ideally, the sets $\{\bm{\alpha}\}$, $\{\bm{\rho}|\tilde{\bm{\alpha}}\}$, and $\{\bm{\tau}|\tilde{\bm{\rho}},\tilde{\bm{\alpha}}\}$ should have closed-form expressions. In this case, we are fortunate that they do.

\section{Examples}\label{sec:examples}

To gain insight into this problem, we examine solutions to several special cases. 

First, we consider $n$-input-port \emph{lossless} systems where $A = O$. This yields simple results:
\begin{align}
\{\bm{\alpha}\} &= \{\bm{0}\}, \\ 
\{\bm{\rho}|\tilde{\bm{\alpha}}=\bm{0}\} &= \mathcal{D}_{n},  \\
\{\bm{\tau}|\tilde{\bm{\rho}}, \tilde{\bm{\alpha}} = \bm{0}\} &= \{\bm{1}-\tilde{\bm{\rho}}^{\uparrow}\},
\end{align}
where
\begin{equation}
\bm{0} = (0,\dots,0), \quad \bm{1} = (1,\dots,1),    
\end{equation}
and $\uparrow$ denotes reordering the vector components in non-decreasing order. Note that $\bm{\alpha} = \bm{0}$ while $\bm{\tau} = \bm{1} - \bm{\rho}^\uparrow$. Thus, this problem is trivial when loss is absent. 

\begin{figure}[htbp]
    \centering
    \includegraphics[width=0.5\textwidth]{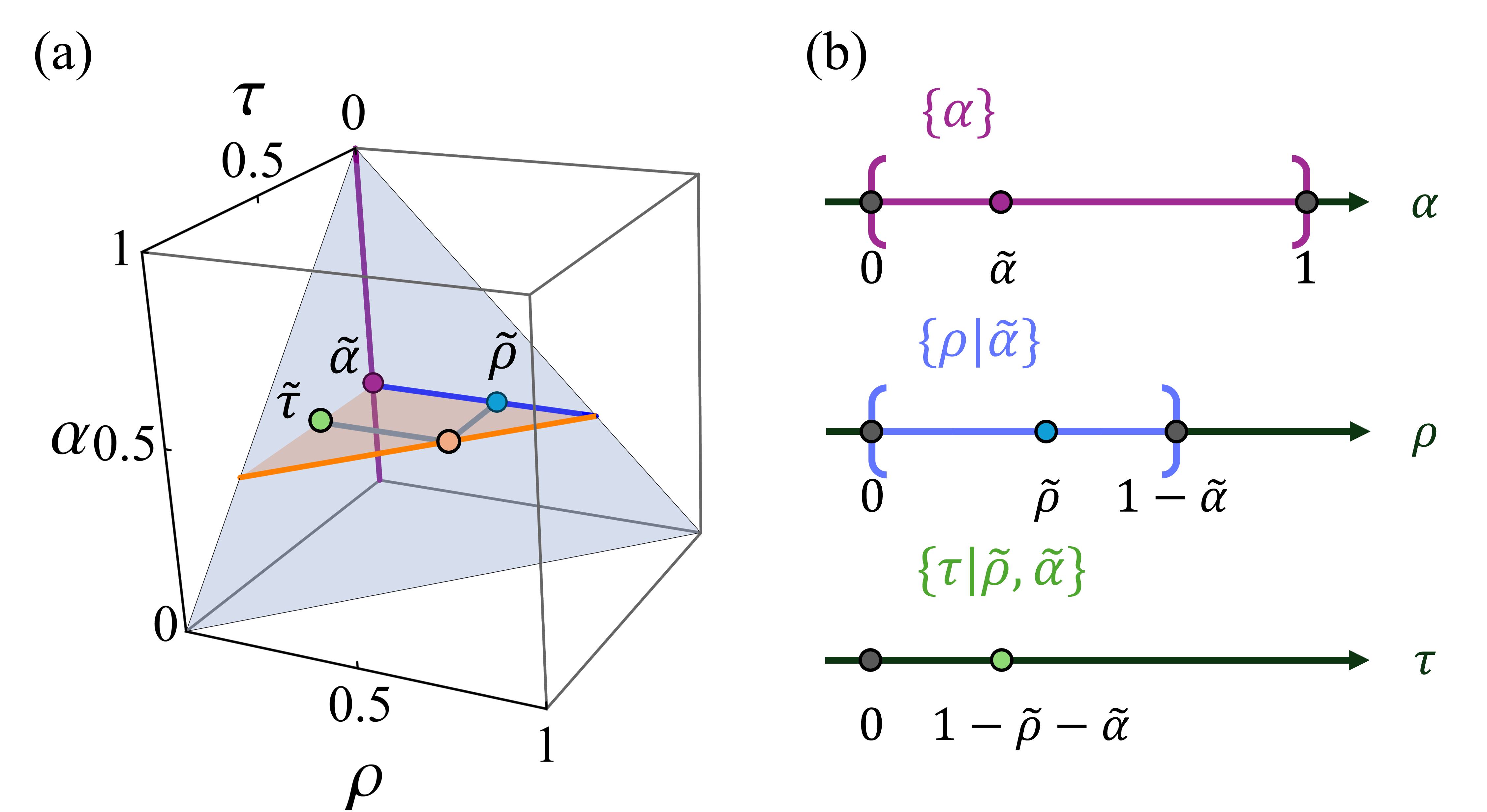}
    \caption{Possible tuples $(\tau, \rho, \alpha)$ for $1$-input-port passive systems. 
    $\Omega_1 =\{(\tilde{\tau}, \tilde{\rho}, \tilde{\alpha}): \tilde{\alpha} \in \{\alpha\},  
    \tilde{\rho} \in \{\rho|\tilde{\alpha}\},  \tilde{\tau} \in \{\tau|\tilde{\rho}, \tilde{\alpha}\}\}$. 
    (a) 3D representation of $\Omega_1$ (blue triangle). 
    (b) Sequential slices of $\Omega_1$ illustrating $\{\alpha\}$, $\{\rho|\tilde{\alpha}\}$, and 
    $\{\tau|\tilde{\rho}, \tilde{\alpha}\}$. Example parameters: $\tilde{\alpha} = 0.3$, 
    $\tilde{\rho} = 0.4$, $\tilde{\tau} = 0.3$.} 
    \label{fig:1-port_passive}
\end{figure}

Next, we examine $1$-input-port passive systems. In this case, $T$, $R$, and $A$ are scalars, leading to straightforward solutions:
\begin{align}
\{\alpha\} &= [0,1], \\
\{\rho|\tilde{\alpha}\} &= [0,1-\tilde{\alpha}],  \\
\{\tau|\tilde{\rho}, \tilde{\alpha}\} &= \{1-\tilde{\rho}-\tilde{\alpha}\}.
\end{align}
Figure \ref{fig:1-port_passive} illustrates these results.

\begin{figure}[htbp]
    \centering
    \includegraphics[width=0.7\textwidth]{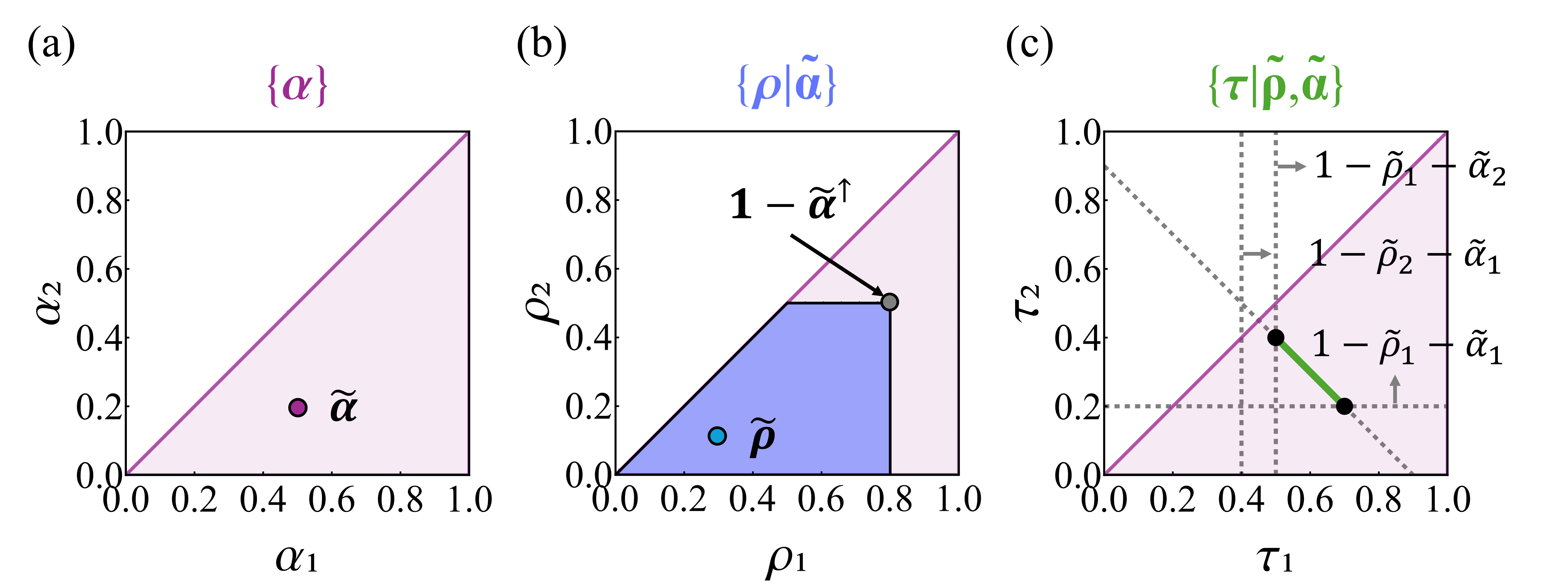}
    \caption{Possible tuples $(\bm{\tau}, \bm{\rho}, \bm{\alpha})$ for $2$-input-port passive systems. Sequential slices of 
    $\Omega_2 =\{(\tilde{\bm{\tau}}, \tilde{\bm{\rho}}, \tilde{\bm{\alpha}}): 
    \tilde{\bm{\alpha}} \in \{\bm{\alpha}\}, \tilde{\bm{\rho}} \in \{\bm{\rho}|\tilde{\bm{\alpha}}\}, 
    \tilde{\bm{\tau}} \in \{\bm{\tau}|\tilde{\bm{\rho}}, \tilde{\bm{\alpha}}\}\}$: (a) $\{\bm{\alpha}\}$, (b) $\{\bm{\rho}|\tilde{\bm{\alpha}}\}$, 
    (c) $\{\bm{\tau}|\tilde{\bm{\rho}}, \tilde{\bm{\alpha}}\}$. 
    Example parameters: $\tilde{\bm{\alpha}} = (0.5,0.2)$, $\tilde{\bm{\rho}} = (0.3,0.1)$.}
    \label{fig:2-port_passive}
\end{figure}

The problem becomes more complicated for passive systems with multiple input ports. For $2$-input-port passive systems, we will prove that:
\begin{align}
\{\bm{\alpha}\} &= \mathcal{D}_{2}, \label{eq:2-port-alpha}\\ 
\{\bm{\rho}|\tilde{\bm{\alpha}}\} &= \{\bm{\rho}\in \mathcal{D}_{2}: \rho_{1} \leq 1-\tilde{\alpha}_{2}, \rho_{2} \leq 1-\tilde{\alpha}_{1}\},  \\
\{\bm{\tau}|\tilde{\bm{\rho}}, \tilde{\bm{\alpha}}\} &= \{\bm{\tau}\in \mathcal{D}_{2}: \\ 
 & \tau_{2} \geq 1 - \tilde{\rho}_{1} - \tilde{\alpha}_{1}, \label{eq:2-port-ineq1}\\
 & \tau_{1} \geq 1 - \tilde{\rho}_{1} - \tilde{\alpha}_{2}, \label{eq:2-port-ineq2}\\
 & \tau_{1} \geq 1 - \tilde{\rho}_{2} - \tilde{\alpha}_{1},  \label{eq:2-port-ineq3}\\
& \tau_{1} + \tau_{2} = 2 - \tilde{\rho}_{1} - \tilde{\rho}_{2} -\tilde{\alpha}_{1} - \tilde{\alpha}_{2}  
\}. \label{eq:2-port-trace-eq}
\end{align}
Figure \ref{fig:2-port_passive} illustrates these results. In general, $\{\bm{\alpha}\}$ forms a right triangle, $\{\bm{\rho}|\tilde{\bm{\alpha}}\}$ a right trapezoid, and $\{\bm{\tau}|\tilde{\bm{\rho}}, \tilde{\bm{\alpha}}\}$ a line segment.

\begin{figure}[hbtp]
    \centering
    \includegraphics[width=0.6\textwidth]{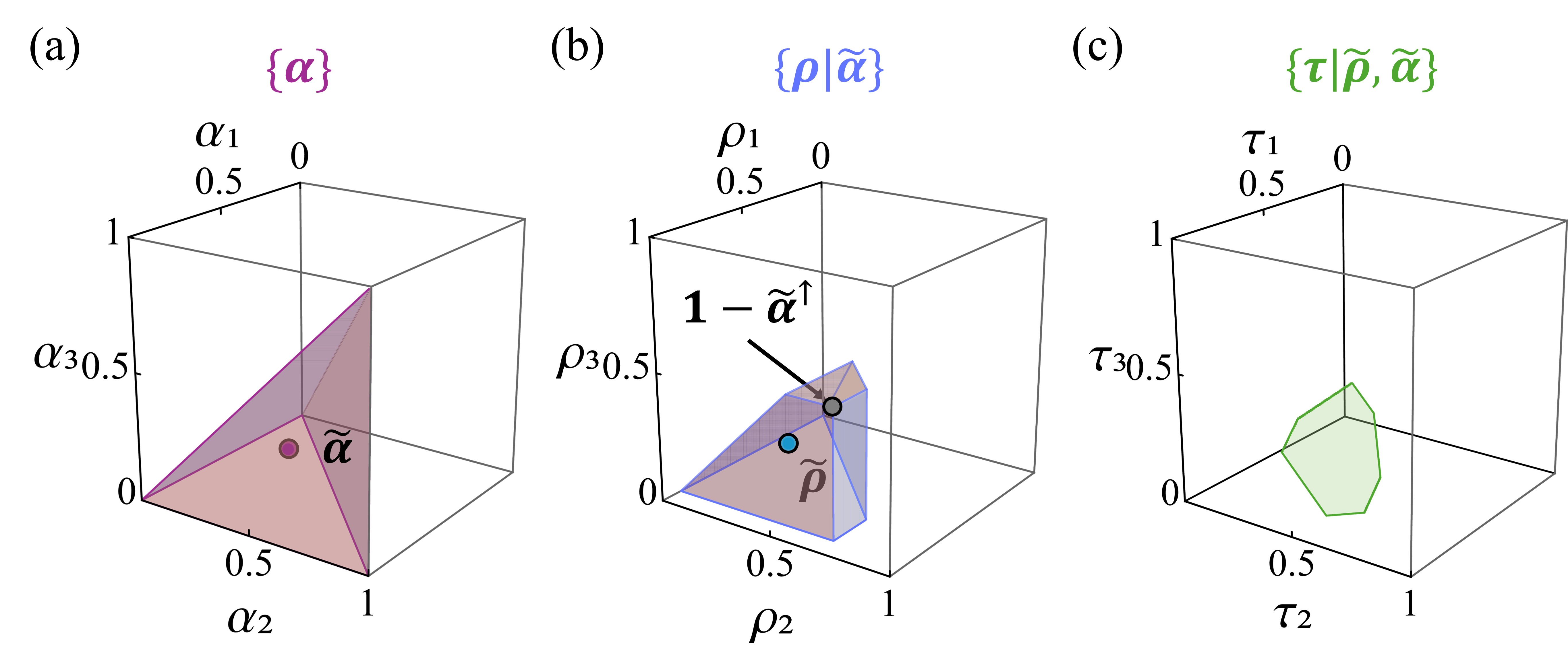}
    \caption{Possible tuples $(\bm{\tau}, \bm{\rho}, \bm{\alpha})$ for $3$-input-port passive systems. Sequential slices of 
    $\Omega_3 =\{(\tilde{\bm{\tau}}, \tilde{\bm{\rho}}, \tilde{\bm{\alpha}}): 
    \tilde{\bm{\alpha}} \in \{\bm{\alpha}\}, \tilde{\bm{\rho}} \in \{\bm{\rho}|\tilde{\bm{\alpha}}\}, 
    \tilde{\bm{\tau}} \in \{\bm{\tau}|\tilde{\bm{\rho}}, \tilde{\bm{\alpha}}\}\}$: (a) $\{\bm{\alpha}\}$, (b) $\{\bm{\rho}|\tilde{\bm{\alpha}}\}$, 
    (c) $\{\bm{\tau}|\tilde{\bm{\rho}}, \tilde{\bm{\alpha}}\}$. 
    Example parameters: $\tilde{\bm{\alpha}} = (0.5,0.3,0.1)$, $\tilde{\bm{\rho}} = (0.5,0.2,0.1)$.}
    \label{fig:3-port_passive}
\end{figure}

For $3$-input-port passive systems, we will prove that:
\begin{align}
\{\bm{\alpha}\} &= \mathcal{D}_{3}, \\ 
\{\bm{\rho}|\tilde{\bm{\alpha}}\} &= \{\bm{\rho} \in \mathcal{D}_{3}: \rho_{1} \leq 1-\tilde{\alpha}_{3}, \rho_{2} \leq 1-\tilde{\alpha}_{2}, \rho_{3} \leq 1-\tilde{\alpha}_{1} \}, \\
\{\bm{\tau}|\tilde{\bm{\rho}}, \tilde{\bm{\alpha}}\} &= \{\bm{\tau}\in \mathcal{D}_{3}: \\ 
& \tau_{3} \geq 1 - \tilde{\rho}_{1} - \tilde{\alpha}_{1}, \label{eq:3-port-Weyl-ineq-first} \\
& \tau_{2} \geq 1 - \tilde{\rho}_{1} - \tilde{\alpha}_{2}, \\
& \tau_{2} \geq 1 - \tilde{\rho}_{2} - \tilde{\alpha}_{1}, \\
 & \tau_{1} \geq 1 - \tilde{\rho}_{1} - \tilde{\alpha}_{3}, \\
 & \tau_{1} \geq 1 - \tilde{\rho}_{2} - \tilde{\alpha}_{2},  \\
& \tau_{1} \geq 1 - \tilde{\rho}_{3} - \tilde{\alpha}_{1}, \label{eq:3-port-Weyl-ineq-last} \\
& \tau_{2} + \tau_{3} \geq 2 - \tilde{\rho}_{1} - \tilde{\rho}_{2} - \tilde{\alpha}_{1} - \tilde{\alpha}_{2}, \label{eq:3-port-Fan-ineq} \\
& \tau_{1} + \tau_{3} \geq 2 - \tilde{\rho}_{1} - \tilde{\rho}_{3} - \tilde{\alpha}_{1} - \tilde{\alpha}_{2}, \label{eq:3-port-Lidskii-ineq-first}\\
& \tau_{1} + \tau_{3} \geq 2 - \tilde{\rho}_{1} - \tilde{\rho}_{2} - \tilde{\alpha}_{1} - \tilde{\alpha}_{3}, \\  
& \tau_{1} + \tau_{2} \geq 2 - \tilde{\rho}_{1} - \tilde{\rho}_{2} - \tilde{\alpha}_{2} - \tilde{\alpha}_{3}, \\
& \tau_{1} + \tau_{2} \geq 2 - \tilde{\rho}_{2} - \tilde{\rho}_{3} - \tilde{\alpha}_{1} - \tilde{\alpha}_{2}, \label{eq:3-port-Lidskii-ineq-last}\\
& \tau_{1} + \tau_{2} \geq 2 - \tilde{\rho}_{1} - \tilde{\rho}_{3} - \tilde{\alpha}_{1} - \tilde{\alpha}_{3},\label{eq:3-port-Horn-ineq} \\
& \tau_{1} + \tau_{2} + \tau_{3} = 3 - \tilde{\rho}_{1} - \tilde{\rho}_{2} - \tilde{\rho}_{3} -\tilde{\alpha}_{1} - \tilde{\alpha}_{2}  -\tilde{\alpha}_{3} 
\}. \label{eq:3-port-trace-eq}
\end{align}
Figure \ref{fig:3-port_passive} illustrates these results. In general, $\{\bm{\alpha}\}$ forms a tetrahedron, $\{\bm{\rho}|\tilde{\bm{\alpha}}\}$ a convex hexahedron, and $\{\bm{\tau}|\tilde{\bm{\rho}}, \tilde{\bm{\alpha}}\}$ a convex polygon lying on the plane defined by the trace equality~(\ref{eq:3-port-trace-eq}).

\section{Main results}~\label{sec:main_results}

The examples in Section~\ref{sec:examples} illustrate the main result of this paper:  
\begin{theorem*}
 For $n$-input-port passive linear time-invariant systems, the set of all attainable tuples of transmission, reflection, and absorption eigenvalues is 
\begin{equation}
\Omega_{n} = \{(\tilde{\bm{\tau}},\tilde{\bm{\rho}},\tilde{\bm{\alpha}}) \in \mathcal{D}_{n} \times \mathcal{D}_{n} \times \mathcal{D}_{n}: \tilde{\bm{\alpha}} \in \{\bm{\alpha}\}, \tilde{\bm{\rho}} \in \{\bm{\rho}|\tilde{\bm{\alpha}}\}, \tilde{\bm{\tau}} \in \{\bm{\tau}|\tilde{\bm{\rho}}, \tilde{\bm{\alpha}}\} \},
\end{equation}
where 
\begin{align}
\{\bm{\alpha}\} = &\mathcal{D}_{n}, \label{eq:result-1}  \\
\{\bm{\rho}|\tilde{\bm{\alpha}}\} = &\{\bm{\rho} \in \mathcal{D}_{n}: \rho_{i} \leq 1 - \tilde{\alpha}_{n+1-i}, \, i=1,\dots,n\}, \label{eq:result-2} \\
\{\bm{\tau}|\tilde{\bm{\rho}},\tilde{\bm{\alpha}}\} = &\{\bm{\tau} \in \mathcal{D}_{n}: \sum_{k=1}^{n} \tau_k = n - \sum_{i=1}^{n} \tilde{\rho}_i - \sum_{j=1}^{n} \tilde{\alpha}_j; \label{eq:main_result_trace}
\\
&\sum_{k\in K} \tau_{n+1-k} \geq r - \sum_{i\in I} \tilde{\rho}_i - \sum_{j\in J} \tilde{\alpha}_j, \quad \forall (I,J,K) \in T_r^n, \; r = 1,2,\ldots,n-1. \}\label{eq:main_result_ineq}
\end{align}
\end{theorem*}
The condition (\ref{eq:main_result_trace}) represents the trace equality~(\ref{eq:trace_equality}). The condition (\ref{eq:main_result_ineq}) comprises a set of linear inequalities known as Horn's inequalities~\cite{horn1962}, where 
\begin{equation}
I = \{i_1 < i_2 < \ldots< i_r\}, \quad J = \{j_1 < j_2 < \ldots< j_r\}, \quad K = \{k_1 < k_2 < \ldots < k_r\}
\end{equation}
are certain ordered subsets of $\{1,2,\ldots,n\}$ with the same cardinality $r \in \{1,2,\ldots,n-1\}$. We define admissible triples $(I,J,K)$ as follows. The set $T^n_r$ of all admissible triples $(I,J,K)$ of cardinality $r$ can be described by induction on $r$~\cite{fulton2000,bhatia2001,niculescu2018}: For $r=1$, $(I,J,K)$ is in $T^{n}_1$ if and only if
\begin{equation}
    i_1 + j_1 = k_1 +1.
\end{equation}
For $r>1$, $(I,J,K)$ is in $T^{n}_r$ if and only if
\begin{equation}
\sum\limits_{u= 1}^{r} i_u + \sum\limits_{v=1}^{r} j_v = \sum\limits_{w=1}^{r} k_w + \frac{r(r+1)}{2},   
\end{equation}
and, for all $1\leq p\leq r-1$ and all $(U,V,W)\in T^{r}_{p}$,
\begin{equation}
\sum\limits_{u\in U} i_u + \sum\limits_{v \in V} j_v \leq \sum\limits_{w \in W} k_{w} + \frac{p(p+1)}{2}.
\end{equation}


\begin{proof}[Proof of Eq. (\ref{eq:result-1})]
We know that $\{\bm{\alpha}\} \subseteq \mathcal{D}_{n}$. To complete the proof, we need to prove that $\mathcal{D}_{n} \subseteq \{\bm{\alpha}\}$. It suffices to show that for any $\bm{x} \in \mathcal{D}_{n}$, there exists a system with absorption eigenvalues $\bm{x}$. We construct the following $n \times n$ transmission and reflection matrices:
\begin{align} 
t &= O,  \\
r &= \operatorname{diag} (\sqrt{1-x_{1}}, \sqrt{1-x_{2}}, \dots, \sqrt{1-x_{n}}),
\end{align}
where $\operatorname{diag}(\bm{v})$ represents a diagonal matrix with the elements of vector $\bm{v}$ along its main diagonal. Then the corresponding transmittance, reflectance, and absorptivity matrices
\begin{align} 
T &= t^\dagger t = O,  \\
R &= r^\dagger r = \operatorname{diag} (1-x_{1}, 1-x_{2}, \dots, 1-x_{n}), \\
A &= I - t^\dagger t - r^\dagger r =\operatorname{diag} (x_{1}, x_{2}, \dots, x_{n}),
\end{align}
satisfy conditions (\ref{eq:condition-psd-T-R}), (\ref{eq:condition-psd}) and (\ref{eq:condition-id}). Moreover, $\bm{\lambda}(A) = \bm{x}$. This confirms that $\mathcal{D}_{n} \subseteq \{\bm{\alpha}\}$ and completes the proof of Eq.~(\ref{eq:result-1}).    
\end{proof}

\begin{proof}[Proof of Eq.~(\ref{eq:result-2})]
Let $\mathcal{P}_{n}$ denote the set on the right-hand side of Eq.~(\ref{eq:result-2}). We first prove that $\{\bm{\rho}|\tilde{\bm{\alpha}}\} \subseteq \mathcal{P}_{n}$. We know that $\{\bm{\rho}|\tilde{\bm{\alpha}}\} \subseteq \mathcal{D}_{n}$. For $\tilde{\bm{\rho}} \in \{\bm{\rho}|\tilde{\bm{\alpha}}\}$, we have
\begin{equation}
\tilde{\bm{\rho}} = \bm{\lambda}(R) = \bm{\lambda}(I-A-T).     
\end{equation}
We now apply Loewner's theorem~\cite{lowner1934}: For $n \times n$ Hermitian matrices $M$ and $N$, if $M - N$ is positive semidefinite, then
\begin{equation}
\lambda_{j}(M) \geq \lambda_{j}(N), \quad j=1,\dots, n.     
\end{equation}
Let $M = I - A$ and $N = I - A - T$, so that $M - N = T$. As $T$ is positive semidefinite, Loewner's theorem implies that
\begin{equation}
\tilde{\rho}_{i} \leq \lambda_{i}(I-A) = 1 - \lambda_{n+1-i}(A) = 1 - \tilde{\alpha}_{n+1-i}.  
\end{equation} 
This confirms that $\tilde{\bm{\rho}} \in \mathcal{P}_{n}$. Next, we prove $\mathcal{P}_{n} \subseteq \{\bm{\rho}|\tilde{\bm{\alpha}}\}$. It suffices to prove that for any $\bm{y} \in \mathcal{P}_{n}$, there exists a system with reflection eigenvalues $\bm{y}$, given prescribed absorption eigenvalues $\tilde{\bm{\alpha}} \in \mathcal{D}_n$. We construct the following $n\times n$ transmission and reflection matrices:
\begin{align} 
t &= \operatorname{diag} (\sqrt{1-y_{1}-\tilde{\alpha}_{1}}, \sqrt{1-y_{2}-\tilde{\alpha}_{2}}, \dots, \sqrt{1-y_{n}-\tilde{\alpha}_{n}}),  \\
r &= \operatorname{diag} (\sqrt{y_{1}}, \sqrt{y_{2}}, \dots, \sqrt{y_{n}}).
\end{align}
Then the corresponding transmittance, reflectance, and absorptivity matrices
\begin{align} 
T &= t^\dagger t = \operatorname{diag} (1-y_{1}-\tilde{\alpha}_{1}, 1-y_{2}-\tilde{\alpha}_{2}, \dots, 1-y_{n}-\tilde{\alpha}_{n}),  \\
R &= r^\dagger r =  \operatorname{diag} (y_{1}, y_{2}, \dots, y_{n}), \\
A &= I - t^\dagger t - r^\dagger r = \operatorname{diag} (\tilde{\alpha}_{1}, \tilde{\alpha}_{2}, \dots, \tilde{\alpha}_{n}),
\end{align}
satisfy conditions (\ref{eq:condition-psd-T-R}), (\ref{eq:condition-psd}) and (\ref{eq:condition-id}). Moreover, $\bm{\lambda}(R) = \bm{y}$ and $\bm{\lambda}(A) = \tilde{\bm{\alpha}}$. This confirms that $\mathcal{P}_{n} \subseteq \{\bm{\rho}|\tilde{\bm{\alpha}}\}$ and completes the proof of Eq.~(\ref{eq:result-2}).
\end{proof}


\begin{proof}[Proof of Eqs.~(\ref{eq:main_result_trace}) and (\ref{eq:main_result_ineq})]
Let $\mathcal{T}_{n}$ denote the set on the right-hand side of Eq.~(\ref{eq:main_result_ineq}). We begin by proving that $\{\bm{\tau}|\tilde{\bm{\rho}}, \tilde{\bm{\alpha}}\} \subseteq \mathcal{T}_{n}$. We have established that $\{\bm{\tau}|\tilde{\bm{\rho}}, \tilde{\bm{\alpha}}\} \subseteq \mathcal{D}_{n}$. For any $\tilde{\bm{\tau}} \in \{\bm{\tau}|\tilde{\bm{\rho}}, \tilde{\bm{\alpha}}\}$, we have
\begin{equation}
\tilde{\bm{\tau}} = \bm{\lambda}(T) = \bm{\lambda}(I-R-A),     
\end{equation}
and thus, 
\begin{equation}\label{eq:tau_lambda_componentwise}
\tilde{\tau}_{n+1-k} = 1 - \lambda_k (R+A). 
\end{equation}
We now determine the possible values of $\bm{\lambda}(R+A)$ given the prescribed eigenvalues $\bm{\lambda}(R)=\tilde{\bm{\rho}}$ and $\bm{\lambda}(A)=\tilde{\bm{\alpha}}$. We apply Alfred Horn's theorem~\cite{klyachko1998,knutson1999,fulton2000}, which states: A triple of non-increasingly ordered real $n$-vectors $(\bm{a},\bm{b},\bm{c})$ represents the eigenvalues of $n\times n$ Hermitian matrices $M,N,$ and $M+N$ if and only if 
\begin{equation}\label{eq:Horn_eq}
\sum_{k=1}^{n}c_{k} = \sum_{i=1}^{n}a_{i} + \sum_{j=1}^{n}b_{j}    
\end{equation}
and
\begin{equation}\label{eq:Horn_ineq}
\sum_{k\in K}c_{k} \leq \sum_{i\in I}a_{i} + \sum_{j\in J}b_{j}    
\end{equation}
for all $(I,J,K)$ in $T^{n}_{r}$, and all $r<n$. We set $M= R$, $N=A$, thus $(\bm{a},\bm{b},\bm{c}) = (\tilde{\bm{\rho}},\tilde{\bm{\alpha}},\bm{\lambda}(R+A))$ satisfies the conditions~(\ref{eq:Horn_eq}) and (\ref{eq:Horn_ineq}). Combining these conditions with Eq.~(\ref{eq:tau_lambda_componentwise}), we obtain
\begin{align}
&\sum_{k=1}^{n} \tilde{\tau}_k = n - \sum_{i=1}^{n} \tilde{\rho}_i - \sum_{j=1}^{n} \tilde{\alpha}_j; \\
&\sum_{k\in K} \tilde{\tau}_{n+1-k} \geq r - \sum_{i\in I} \tilde{\rho}_i - \sum_{j\in J} \tilde{\alpha}_j,
\end{align}
for all $(I,J,K)$ in $T^{n}_{r}$, and all $r<n$. This confirms that $\tilde{\bm{\tau}} \in \mathcal{T}_n$ and thus $\{\bm{\tau}|\tilde{\bm{\rho}}, \tilde{\bm{\alpha}}\} \subseteq \mathcal{T}_{n}$.

Next, we demonstrate that $\mathcal{T}_{n} \subseteq \{\bm{\tau}|\tilde{\bm{\rho}},\tilde{\bm{\alpha}}\}$. It suffices to prove that for any $\bm{z} \in \mathcal{T}_{n}$, there exists a system with transmission eigenvalues $\bm{z}$, given prescribed reflection eigenvalues $\tilde{\bm{\rho}} \in \mathcal{P}_n$ and absorption eigenvalues $\tilde{\bm{\alpha}} \in \mathcal{D}_n$. Observe that $\bm{z} \in \mathcal{T}_n$ implies the triple $(\tilde{\bm{\rho}}, \tilde{\bm{\alpha}}, \bm{1}-\bm{z}^\uparrow)$ satisfies conditions~(\ref{eq:Horn_eq}) and (\ref{eq:Horn_ineq}). By Alfred Horn's theorem, there exist two $n\times n$ Hermitian matrices $R$ and $A$ satisfying
\begin{equation}
\bm{\lambda}(R) = \tilde{\bm{\rho}}, \quad \bm{\lambda}(A) = \tilde{\bm{\alpha}}, \quad  \bm{\lambda}(R+A) = \bm{1}-\bm{z}^\uparrow.
\end{equation}
We then define 
\begin{equation}
    T = I - R - A.
\end{equation}
By construction, $T$, $R$, and $A$ satisfy condition (\ref{eq:condition-id}), and
\begin{equation}
\bm{\lambda}(T) =  \bm{z}.
\end{equation}
Moreover, since $\bm{z} \in \mathcal{T}_n \subseteq \mathcal{D}_n$, $\tilde{\bm{\rho}} \in \mathcal{P}_n \subseteq \mathcal{D}_n$, and $\tilde{\bm{\alpha}} \in \mathcal{D}_n$, matrices $T$, $R$, and $A$ are positive semidefinite, satisfying conditions (\ref{eq:condition-psd-T-R}) and (\ref{eq:condition-psd}). For any $n\times n$ positive semidefinite matrix $M$, there exists a unique $n\times n$ positive semidefinite matrix $N$ such that $N^2 = M$ (see Ref.~\cite{zhang2011}, p.~81, Theorem~3.5). We denote this matrix $N$ as the square root of $M$, written as $N = M^{1/2}$. Consequently, we can construct a system with the following $n \times n$ transmission and reflection matrices 
\begin{equation}
t = T^{1/2}, \qquad r = R^{1/2}.
\end{equation}
Such a system satisfies all required conditions. This confirms that $\mathcal{T}_{n} \subseteq \{\bm{\tau}|\tilde{\bm{\rho}},\tilde{\bm{\alpha}}\}$ and completes the proof of Eqs.~(\ref{eq:main_result_trace}) and (\ref{eq:main_result_ineq})  
\end{proof}



\section{Horn's inequalities}\label{sec:discussions_Horn_ineq}

The emergence of Horn's inequalities (\ref{eq:main_result_ineq}) from the passivity condition is a striking result. These inequalities, conjectured by Alfred Horn in 1962 \cite{horn1962}, address Weyl's problem posed in 1912~\cite{weyl1912}: determining the relationships between eigenvalues of Hermitian matrices $M$, $N$, and $M+N$. This long-standing conjecture remained unresolved for nearly four decades.

This conjecture is now a theorem due to the groundbreaking works by Alexander A. Klyachko \cite{klyachko1998} in 1998 and Allen Knutson and Terence Tao \cite{knutson1999} in 1999. The proof utilizes innovative techniques from algebraic geometry and representation theory. An accessible introduction to Horn's inequalities is offered by Rajendra Bhatia~\cite{bhatia2001}. A more technical description of the proof and its implications can be found in William Fulton's review \cite{fulton2000}.

The appearance of Horn's inequalities in our analysis of passivity constraints on transport eigenvalues establishes an unexpected bridge between fundamental wave physics and advanced matrix theory. This connection suggests that the mathematical structures underlying passivity in wave systems are more intricate and far-reaching than previously recognized.

In our context, Horn's inequalities have three key implications:
\begin{enumerate}[(i)]
    \item For given absorption eigenvalues $\tilde{\bm{\alpha}} \in \{\bm{\alpha} \}$ and reflection eigenvalues $\tilde{\bm{\rho}} \in \{\bm{\rho} | \tilde{\bm{\alpha}}\}$, the set of all possible transmission eigenvalues $\{\bm{\tau}|\tilde{\bm{\rho}},\tilde{\bm{\alpha}}\}$ forms a convex polyhedron in $\mathcal{D}_n$.  
    \item This convex polyhedron is characterized by the trace equality and Horn's inequalities.
    \item These inequalities are derivable through an inductive procedure.
\end{enumerate}

The number of Horn's inequalities grows rapidly with increasing $n$. Table~\ref{table:Horn_inequalities} presents the number of Horn's inequalities, denoted as $\verb|#| T^n \coloneqq \sum_r \verb|#| T^n_r$, as a function of $n$, where $\verb|#|$ represents the cardinality of a set. 

\begin{table}[hbtp]
\centering
\caption{The number of Horn's inequalities}
\begin{tabular}{@{}P{1.5cm}P{1.5cm}P{1.5cm}P{1.5cm}P{1.5cm}P{1.5cm}P{1.5cm}P{1.5cm}P{1.5cm}@{}}
\addlinespace 
\toprule
$n$     & $2$ & $3$  & $4$  & $5$   & $6$   & $7$    & $8$    & $\cdots$\\ \midrule
$\verb|#| T^n$ & $3$ & $12$ & $41$ & $142$ & $522$ & $2062$ & $8752$ & $\cdots$\\ \bottomrule
\end{tabular}
\label{table:Horn_inequalities}
\end{table}

Horn's inequalities and the trace equality provide a complete set of conditions characterizing the set $\{\bm{\tau}|\tilde{\bm{\rho}},\tilde{\bm{\alpha}}\}$ in $\mathcal{D}_{n}$. For $n \leq 5$, these inequalities are independent. However, when $n > 5$, some inequalities may become redundant~\cite{knutson1999}. An algorithm based on Schubert calculus can be used to identify the independent inequalities~\cite{fulton2000}.
 
The complete set of Horn's inequalities is complicated and extensive. Therefore, it is beneficial to examine several key subsets of Horn's inequalities. These families of inequalities, valid for all $n$, provide simpler necessary conditions for characterizing $\{\bm{\tau}|\tilde{\bm{\rho}},\tilde{\bm{\alpha}}\}$ in $\mathcal{D}_{n}$. 

\begin{enumerate}[(i)]
    \item \textbf{Weyl's inequalities~\cite{weyl1912}:}
    \begin{equation}
        \tau_{n+2-i-j} \geq 1 - \tilde{\rho}_i - \tilde{\alpha}_j, \quad \text{for}\;\; n+2-i-j \geq 1.
    \end{equation}
    \item  \textbf{Ky Fan's inequalities~\cite{fan1949}:}
    \begin{equation}
        \sum_{k=1}^{r} \tau_{n+1-k} \geq r - \sum_{i=1}^{r} \tilde{\rho}_i -\sum_{j=1}^{r} \tilde{\alpha}_j, \quad \text{for}\;\; 1 \leq r \leq n-1.
    \end{equation}
    \item \textbf{Lidskii–Wielandt inequalities~\cite{lidskii1950,wielandt1955,li1999a}:}
    
    Let $1 \leq r \leq n-1$ and $1 \leq i_1 < i_2 < \ldots < i_r \leq n$. Then
    \begin{align}
        \sum_{k=1}^{r} \tau_{n+1-i_k} \geq r - \sum_{k=1}^{r} \tilde{\rho}_{i_k} -\sum_{j=1}^{r} \tilde{\alpha}_j, \\
        \sum_{k=1}^{r} \tau_{n+1-i_k} \geq r -\sum_{j=1}^{r} \tilde{\rho}_j -\sum_{k=1}^{r} \tilde{\alpha}_{i_k} .
    \end{align}
\end{enumerate}

To illustrate these inequalities, we revisit our previous examples of $2$- and $3$-input-port passive systems. When $n=2$, Weyl's inequalities reduce to inequalities~(\ref{eq:2-port-ineq1}), (\ref{eq:2-port-ineq2}), and (\ref{eq:2-port-ineq3}), which together with the trace equality~(\ref{eq:2-port-trace-eq}) fully characterize $\{\bm{\tau}|\tilde{\bm{\rho}},\tilde{\bm{\alpha}}\}$ in $\mathcal{D}_{n}$.

When $n=3$, Weyl's inequalities correspond to inequalities~(\ref{eq:3-port-Weyl-ineq-first})-(\ref{eq:3-port-Weyl-ineq-last}). Ky Fan's inequalities lead to an additional inequality~(\ref{eq:3-port-Fan-ineq}). Lidskii–Wielandt inequalities yield inequalities~(\ref{eq:3-port-Lidskii-ineq-first})-(\ref{eq:3-port-Lidskii-ineq-last}). However, these inequalities do not form a complete set; one additional inequality~(\ref{eq:3-port-Horn-ineq}) is required. These twelve inequalities, together with the trace equality~(\ref{eq:3-port-trace-eq}), provide a complete characterization of $\{\bm{\tau}|\tilde{\bm{\rho}},\tilde{\bm{\alpha}}\}$ in $\mathcal{D}_{n}$.

\section{Applications}\label{sec:applications}

Our theory provides a powerful framework for analyzing wave transport phenomena. Here, we apply our results to gain new insights into several important physical effects that lead to extremal responses of transmission, reflection, and absorption.

\subsection{Open channels, coherent perfect reflection, and coherent perfect absorption}

We first consider a group of fundamental phenomena in wave transport: open channels, coherent perfect reflection, and coherent perfect absorption. These effects are characterized by the existence of unity transmission, reflection, and absorption eigenvalues, respectively. More specifically, for an $n$-input-port passive linear time-invariant system:
\begin{itemize}
    \item $k$-fold open channels occur when $\tau_1 = \tau_2 = \cdots = \tau_k = 1$.
    \item $k$-fold coherent perfect reflection occurs when $\rho_1 = \rho_2 = \cdots = \rho_k = 1$.
    \item $k$-fold coherent perfect absorption occurs when $\alpha_1 = \alpha_2 = \cdots = \alpha_k = 1$.
\end{itemize}

Our goal is to determine the permissible combinations of transmission, reflection, and absorption eigenvalues for systems exhibiting these phenomena. To illustrate our approach, we focus on the case of $k$-fold coherent perfect absorption. The analysis for open channels and coherent perfect reflection follows a similar pattern due to the symmetry of the problem under the permutation of $T$, $R$, and $A$. 

\subsubsection{Example: 2-input-port passive systems with 1-fold coherent perfect absorption}

Let us examine 2-input-port passive systems with 1-fold coherent perfect absorption (1-CPA). This condition imposes an additional constraint $\alpha_1 = 1$, which further restricts the set $\Omega_2$ given by Eqs.~(\ref{eq:2-port-alpha})-(\ref{eq:2-port-trace-eq}). We obtain the refined set of permissible tuples of eigenvalues:
\begin{align}
\Omega_{2}|_{\text{1-CPA}}  = \{(\tilde{\bm{\tau}},\tilde{\bm{\rho}},\tilde{\bm{\alpha}}) &\in \Omega_{2}: \tilde{\alpha}_1 = 1\}\\
= \{(\tilde{\bm{\tau}},\tilde{\bm{\rho}},\tilde{\bm{\alpha}}) &\in \mathcal{D}_{2} \times \mathcal{D}_{2} \times \mathcal{D}_{2}: \tilde{\bm{\alpha}} \in \{\bm{\alpha}\}, \tilde{\bm{\rho}} \in \{\bm{\rho}|\tilde{\bm{\alpha}}\}, \tilde{\bm{\tau}} \in \{\bm{\tau}|\tilde{\bm{\rho}},
\tilde{\bm{\alpha}}\},  \nonumber \\
    \{\bm{\alpha}\} &= \{(1,\alpha_2): \alpha_2 \in [0,1]\}, \nonumber \\
    \{\bm{\rho}|\tilde{\bm{\alpha}}\} &= \{(\rho_1,0): \rho_1 \in [0,1-\tilde{\alpha}_2]\},  \nonumber \\
    \{\bm{\tau}|\tilde{\bm{\rho}},\tilde{\bm{\alpha}}\} &= \{(1-\tilde{\rho}_1-\tilde{\alpha}_2, 0)\} \quad \}. \label{eq:1CPA}
\end{align}

\begin{figure}[htbp]
    \centering
    \includegraphics[width=0.6\textwidth]{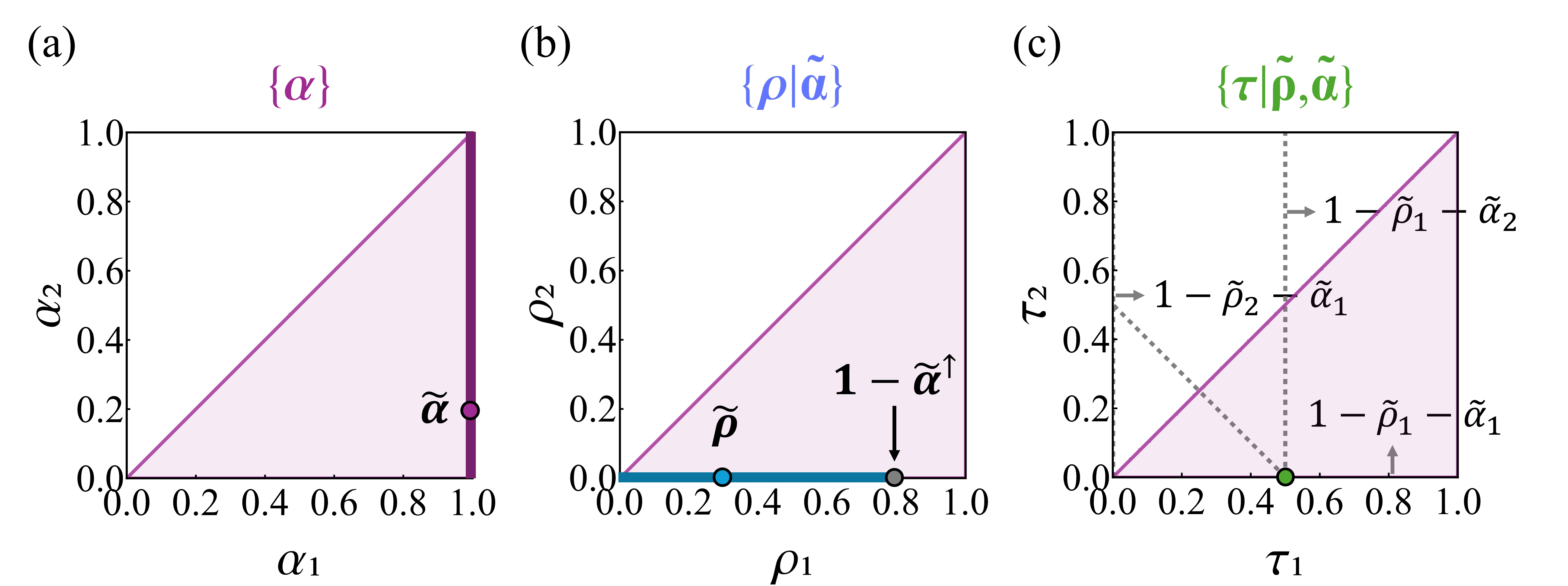}
    \caption{Possible tuples $(\bm{\tau}, \bm{\rho}, \bm{\alpha})$ for 2-input-port passive systems with 1-CPA. (a) The set $\{\bm{\alpha}\}$ (purple interval). (b) The set $\{\bm{\rho}|\tilde{\bm{\alpha}}\}$ for $\tilde{\alpha}_2 = 0.2$ (blue interval). (c) The set $\{\bm{\tau}|\tilde{\bm{\rho}},\tilde{\bm{\alpha}}\}$ for $\tilde{\alpha}_2 = 0.2$ and $\tilde{\rho}_1 = 0.3$ (green point).}
    \label{fig:1CPA-2port}
\end{figure}

Figure \ref{fig:1CPA-2port} illustrates these results. Comparing this with Fig.~\ref{fig:2-port_passive}, we observe that the 1-CPA condition effectively reduces the dimensionality of the sets $\{\bm{\alpha}\}$, $\{\bm{\rho}|\tilde{\bm{\alpha}}\}$, and $\{\bm{\tau}|\tilde{\bm{\rho}},\tilde{\bm{\alpha}}\}$ by 1. Specifically, $\alpha_1 = 1$ implies $\rho_2 = \tau_2 = 0$, while the remaining eigenvalues $(\alpha_2, \rho_1, \tau_1)$ can take values in the entire unrestricted set $\Omega_1$ (see Fig.~\ref{fig:1-port_passive}). Thus, Eq.~(\ref{eq:1CPA}) can be simplified as
\begin{equation}
\begin{split}
\Omega_2|_{\text{1-CPA}} = \{&(\tilde{\bm{\tau}},\tilde{\bm{\rho}},\tilde{\bm{\alpha}}) \in \mathcal{D}_2 \times \mathcal{D}_2 \times \mathcal{D}_2: \tilde{\bm{\alpha}} = (1,\alpha_2), \\
& \tilde{\bm{\rho}} = (\rho_1, 0), \,
\tilde{\bm{\tau}} = (\tau_1, 0), \,(\tau_1, \rho_1, \alpha_2) \in \Omega_{1}\},
\end{split}
\label{eq:Omega_2_1-CPA-simplified}
\end{equation}

\subsubsection{General case: n-input-port passive systems with k-fold coherent perfect absorption}

We can generalize the observation of the previous section to arbitrary $n$-input-port passive systems with $k$-fold coherent perfect absorption ($k$-CPA). The refined set of permissible eigenvalues is:
\begin{equation}
\begin{split}
\Omega_n|_{\text{k-CPA}} = \{&(\tilde{\bm{\tau}},\tilde{\bm{\rho}},\tilde{\bm{\alpha}}) \in \mathcal{D}_n \times \mathcal{D}_n \times \mathcal{D}_n: \tilde{\bm{\alpha}} = (\bm{1}_k, \tilde{\bm{\alpha}}_{[n-k]}), \\
&\tilde{\bm{\rho}} = (\tilde{\bm{\rho}}_{[n-k]}, \bm{0}_k), \,
\tilde{\bm{\tau}} = (\tilde{\bm{\tau}}_{[n-k]}, \bm{0}_k), \, (\tilde{\bm{\tau}}_{[n-k]},\tilde{\bm{\rho}}_{[n-k]},\tilde{\bm{\alpha}}_{[n-k]}) \in \Omega_{n-k}\},
\end{split}
\label{eq:Omega_n_k-CPA}
\end{equation}
where we have introduced the notations
\begin{align}
    &\bm{1}_k := (1,1,\dots,1), \quad \bm{0}_k := (0,0,\dots,0),  \\
    &\tilde{\bm{\alpha}}_{[n-k]} := (\tilde{\alpha}_{k+1},\dots,\tilde{\alpha}_n), \quad \tilde{\bm{\rho}}_{[n-k]} := (\tilde{\rho}_1,\dots,\tilde{\rho}_{n-k}), \quad \tilde{\bm{\tau}}_{[n-k]} := (\tilde{\tau}_1,\dots,\tilde{\tau}_{n-k}).
\end{align}

See Appendix~\ref{appendix:proof_k_CPA} for detailed proof of Eq.~(\ref{eq:Omega_n_k-CPA}). This result reveals a fundamental property of $k$-CPA systems: the $k$ channels with perfect absorption effectively decouple from the remaining $(n-k)$ channels. The eigenvalues corresponding to these $(n-k)$ channels behave as if they belonged to a separate $(n-k)$-input-port system, subject to the full set of constraints derived in Section~\ref{sec:main_results}.

\subsection{Closed channels, reflectionless scattering modes, and coherent zero absorption}

We now consider another group of important transport phenomena: closed channels, reflectionless scattering modes, and coherent zero absorption. These effects are characterized by the existence of zero transmission, reflection, and absorption eigenvalues, respectively. More specifically, for an $n$-input-port passive linear time-invariant system:
\begin{itemize}
    \item $k$-fold closed channels occur when $\tau_{n-k+1} = \tau_{n-k+2} = \cdots = \tau_n = 0$.
    \item $k$-fold reflectionless scattering modes occur when $\rho_{n-k+1} = \rho_{n-k+2} = \cdots = \rho_n = 0$.
    \item $k$-fold coherent zero absorption occurs when $\alpha_{n-k+1} = \alpha_{n-k+2} = \cdots = \alpha_{n} = 0$.
\end{itemize}

Our goal is to determine the permissible combinations of transmission, reflection, and absorption eigenvalues for systems exhibiting these phenomena. To illustrate our approach, we focus on the case of $k$-fold coherent zero absorption. The analysis for closed channels and reflectionless scattering modes follows a similar pattern due to the symmetry of the problem under the permutation of $T$, $R$, and $A$.

\subsubsection{Example: 2-input-port passive systems with 1-fold coherent zero absorption}

Let us examine 2-input-port passive systems with 1-fold coherent zero absorption (1-CZA). This condition imposes an additional constraint $\alpha_2 = 0$, which further restricts the set $\Omega_2$ given by Eqs.~(\ref{eq:result-1})-(\ref{eq:main_result_ineq}). We obtain the refined set of permissible tuples of eigenvalues:
\begin{align}
\Omega_2|_{\text{1-CZA}} = \{(\tilde{\bm{\tau}},\tilde{\bm{\rho}},\tilde{\bm{\alpha}}) &\in \Omega_2: \alpha_2 = 0 \} \nonumber \\
= \{(\tilde{\bm{\tau}},\tilde{\bm{\rho}},\tilde{\bm{\alpha}}) &\in \mathcal{D}_{2} \times \mathcal{D}_{2} \times \mathcal{D}_{2}: \tilde{\bm{\alpha}} \in \{\bm{\alpha}\}, \tilde{\bm{\rho}} \in \{\bm{\rho}|\tilde{\bm{\alpha}}\}, \tilde{\bm{\tau}} \in \{\bm{\tau}|\tilde{\bm{\rho}},
\tilde{\bm{\alpha}}\},  \\
    \{\bm{\alpha}\} &= \{(\alpha_1,0): \alpha_1 \in [0,1]\}, \nonumber \\
  \{\bm{\rho}|\tilde{\bm{\alpha}}\} &= \{(\rho_1,\rho_2) \in \mathcal{D}_{2}: \rho_2 \leq 1 - \tilde{\alpha}_1\}, \nonumber \\
    \{\bm{\tau}|\tilde{\bm{\rho}},\tilde{\bm{\alpha}}\} &= \{(\tau_1, \tau_2) \in \mathcal{D}_{2}: \nonumber  \\  
 & \tau_{2} \geq 1 - \tilde{\rho}_{1} - \tilde{\alpha}_{1}, \nonumber \\
 & \tau_{1} \geq 1 - \tilde{\rho}_{1} \nonumber \\
 & \tau_{1} \geq 1 - \tilde{\rho}_{2} - \tilde{\alpha}_{1},  \nonumber \\
 & \tau_{1} + \tau_{2} = 2 - \tilde{\rho}_{1} - \tilde{\rho}_{2} -\tilde{\alpha}_{1}  \} \quad \}. \label{eq:1CZA}
\end{align}

\begin{figure}[htbp]
    \centering
    \includegraphics[width=0.6\textwidth]{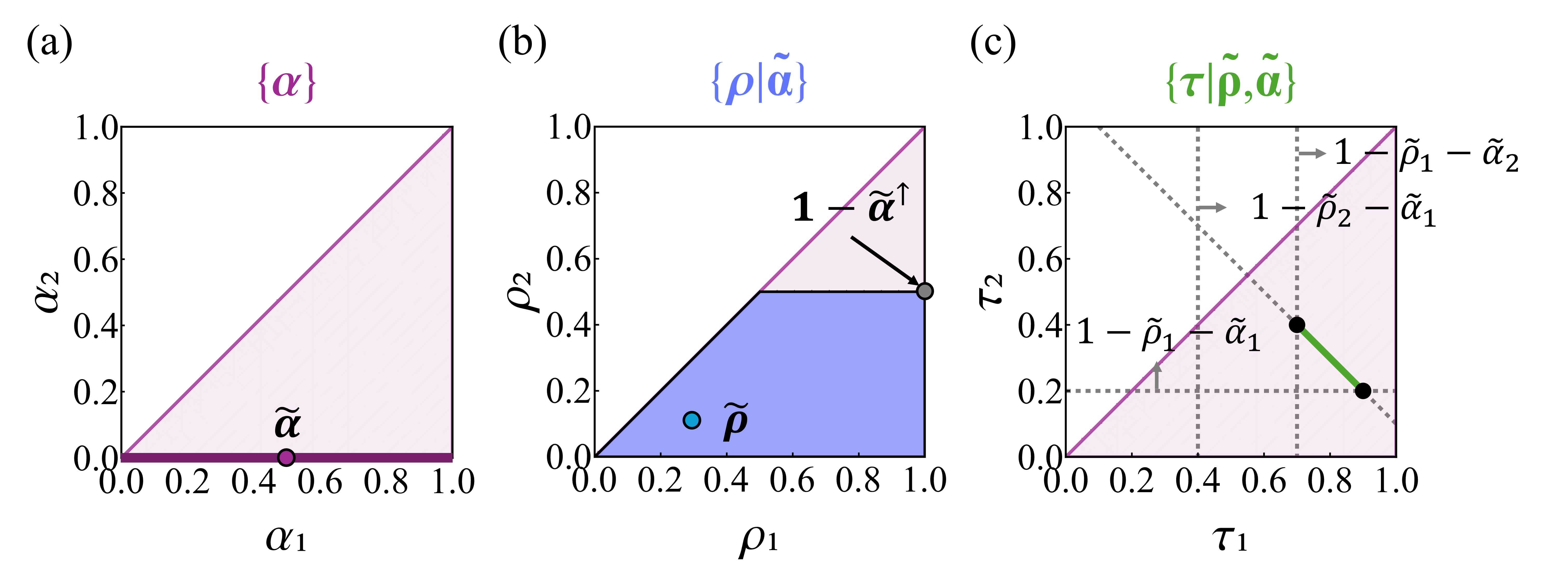}
    \caption{Possible tuples $(\bm{\tau}, \bm{\rho}, \bm{\alpha})$ for 2-input-port passive systems with 1-CZA. (a) The set $\{\bm{\alpha}\}$ (purple interval). (b) The set $\{\bm{\rho}|\tilde{\bm{\alpha}}\}$ for $\tilde{\alpha}_1 = 0.5$ (blue region). (c) The set $\{\bm{\tau}|\tilde{\bm{\rho}},\tilde{\bm{\alpha}}\}$ for $\tilde{\alpha}_1 = 0.5$,  $\tilde{\rho}_1 = 0.3$, and $\tilde{\rho}_2 = 0.1$ (green interval).}
    \label{fig:1CZA-2port}
\end{figure}
Figure \ref{fig:1CZA-2port} illustrates these results. Comparing this with Fig.~\ref{fig:2-port_passive}, we observe that the 1-CZA condition only reduces the dimensionality of the set $\{\bm{\alpha}\}$ by 1. This is in sharp contrast to the 1-CPA case which reduces the dimensionality of all three sets $\{\bm{\alpha}\}$, $ \{\bm{\rho}|\tilde{\bm{\alpha}}\}$, and $\{\bm{\tau}|\tilde{\bm{\rho}},\tilde{\bm{\alpha}}\}$ by 1. Consequently, Eq.~(\ref{eq:1CZA}) cannot be further simplified into a form like Eq.~(\ref{eq:Omega_2_1-CPA-simplified}) for the 1-CPA case.

\subsubsection{General case: n-input-port passive systems with k-fold coherent zero absorption}

We can generalize the observation of the previous section to arbitrary $n$-input-port passive systems with $k$-fold coherent zero absorption ($k$-CZA). The refined set of permissible eigenvalues is:
\begin{equation}
\Omega_n|_{\text{k-CZA}} = \{(\tilde{\bm{\tau}},\tilde{\bm{\rho}},\tilde{\bm{\alpha}}) \in \Omega_n: \tilde{\alpha}_{n-k+1} = \tilde{\alpha}_{n-k+2} = \cdots = \tilde{\alpha}_{n} = 0\},
\label{eq:Omega_n_k-CZA}
\end{equation}
which cannot be further simplified into a form like Eq.~(\ref{eq:Omega_n_k-CPA}) for the $k$-CPA case.

This result reveals a fundamental difference between $k$-CZA and $k$-CPA systems: the $k$ channels with zero absorption cannot be completely decoupled from the remaining $(n-k)$ channels. This difference arises because, in passive systems, zero absorption only implies that the sum of transmission and reflection reaches unity, but this imposes no direct constraints on individual quantities. In contrast, unity absorption implies that both transmission and reflection must reach zero individually.

\section{Discussion and Conclusion}\label{sec:conclusion}

This work provides a comprehensive characterization of passivity constraints on transmission, reflection, and absorption eigenvalues in linear time-invariant systems. Our analysis reveals that the set of permissible eigenvalue combinations forms a convex polyhedron in eigenvalue space, fully described by a trace equality and a set of linear inequalities including Horn's inequalities. This connection between fundamental physics and advanced matrix theory provides new insights into wave transport phenomena.

Our findings apply to all wave types, including optical and acoustic varieties, with significant implications for both theoretical studies and practical applications. From a theoretical perspective, our results provide a new lens through which to view the interplay between passivity and transport properties. The emergence of Horn's inequalities in this context suggests deeper mathematical structures underlying wave physics that warrant further exploration.

In terms of practical applications, the constraints we've derived can guide the design and optimization of passive wave devices across a wide range of fields, including optics, acoustics, and mesoscopic physics. The convex nature of the allowable eigenvalue set is particularly useful for convex optimization techniques~\cite{boyd2004}, potentially informing the development of more efficient photonic and acoustic components.

Several important directions for future research emerge from this work:

1. Symmetry considerations: We have not assumed any additional symmetries of the passive systems. An interesting direction for future research would be to study the implications of additional symmetries~\cite{guo2022c} on the attainable range of transport eigenvalues.

2. Statistical theory: Our results determine the support of the joint probability distribution for the three sets of eigenvalues. This paves the way for a more comprehensive random-matrix theory on these transport eigenvalues, potentially leading to new statistical insights into wave transport in passive systems.

3. Optimization algorithms: Developing efficient algorithms to find optimal eigenvalue combinations within the constrained set for various objectives could be valuable for device design and optimization.

4. Experimental validation: While our theory is mathematically rigorous, experimental validation across different wave systems would further solidify its practical relevance.

In conclusion, our work provides a powerful toolset for analyzing and designing passive wave transport systems. As we continue to push the boundaries of our understanding of wave phenomena, we anticipate that the framework presented here will play a crucial role in future advancements across a wide range of fields, from fundamental physics to cutting-edge technologies.

\begin{acknowledgments}
This work is funded by the U.S. Department of Energy (Grant No. DE-FG02-07ER46426) and by a Simons Investigator in Physics grant from the Simons Foundation (Grant No. 827065).
\end{acknowledgments}

\appendix

\appendix

\section{Detailed Proof of Eq.~(\ref{eq:Omega_n_k-CPA})}\label{appendix:proof_k_CPA}

\begin{proof}
We prove that passivity and the $k$-CPA condition together imply that, under a certain unitary similarity transformation, the $A$, $R$, and $T$ matrices can be expressed in the following block form:
\begin{equation}
UAU^{\dagger} = \begin{pmatrix}
I_k & O \\
O & A_{[n-k]}
\end{pmatrix}, \quad  
URU^{\dagger} = \begin{pmatrix}
O_k & O \\
O & R_{[n-k]}
\end{pmatrix}, \quad 
UTU^{\dagger} = \begin{pmatrix}
O_k & O \\
O & T_{[n-k]}
\end{pmatrix}.
\label{eq:block-form}
\end{equation}
Here, $U$ is an $n\times n$ unitary matrix, $I_k$ and $O_k$ are the $k\times k$ identity and zero matrices respectively, $A_{[n-k]}$, $R_{[n-k]}$, and $T_{[n-k]}$ are $(n-k)\times (n-k)$ matrices, and $O$ denotes zero matrices of appropriate sizes.

By the $k$-CPA condition, $A$ has $k$ eigenvalues equal to 1:
\begin{equation}
\tilde{\alpha}_1 = \tilde{\alpha}_2 = \cdots = \tilde{\alpha}_k = 1,
\label{eq:k-CPA-condition}
\end{equation}
with corresponding orthonormal eigenvectors $\ket{\psi_1}, \ket{\psi_2}, \ldots, \ket{\psi_k}$. We extend these to form an orthonormal basis $\{\ket{\psi_i}\}_{i=1}^n$ of the entire space. Let $\{\ket{e_i}\}_{i=1}^n$ be the standard basis. We define the unitary matrix $U = [U_{ij}]$ by:
\begin{equation}
U_{ij} = \braket{\psi_i|e_j}.
\label{eq:U-definition}
\end{equation}
Under the unitary similarity transformation induced by $U$, $A$ takes the block form:
\begin{equation}
UAU^{\dagger} = \begin{pmatrix}
I_k & B \\
B^{\dagger} & A_{[n-k]}
\end{pmatrix},
\label{eq:A-block-form}
\end{equation}
where $B$ is a $k\times (n-k)$ complex matrix. We now prove that $B = O$. Consider the matrix:
\begin{equation}
U(I-A)U^{\dagger} = \begin{pmatrix}
O_k & -B \\
-B^{\dagger} & I_{n-k} - A_{[n-k]}
\end{pmatrix}.
\label{eq:I-A-block-form}
\end{equation}
Since $I-A = T + R$ and both $T$ and $R$ are positive semidefinite by passivity, we have $I-A \geq O$. This implies $U(I-A)U^{\dagger} \geq O$. For a positive semidefinite matrix $H = [h_{ij}]$, if $h_{kk} = 0$ for some $k$, then $h_{ik} = h_{ki} = 0$ for all $i$ (see Ref.~\cite{horn2017}, p.432, Observation 7.1.10). Applying this to $U(I-A)U^{\dagger}$, we conclude that $B = O$. Thus, $A$ takes the form in Eq.~(\ref{eq:block-form}). 

Now, consider $T$ and $R$. We can express $UTU^{\dagger}$ and $URU^{\dagger}$ in block form:
\begin{equation}
UTU^{\dagger} = \begin{pmatrix}
T_{[k]} & C \\
C^{\dagger} & T_{[n-k]}
\end{pmatrix}, \quad 
URU^{\dagger} = \begin{pmatrix}
R_{[k]} & D \\
D^{\dagger} & R_{[n-k]}
\end{pmatrix}.
\label{eq:TR-block-form}
\end{equation}
The passivity condition implies $T_{[k]} \geq O$ and $R_{[k]} \geq O$. From Eq.~(\ref{eq:I-A-block-form}), we have:
\begin{equation}
U(T+R)U^{\dagger} = U(I-A)U^{\dagger} = \begin{pmatrix}
O_k & O \\
O & I_{n-k} - A_{[n-k]}
\end{pmatrix}.
\label{eq:TR-sum-block-form}
\end{equation}
This implies $T_{[k]} + R_{[k]} = O_k$. Given that $T_{[k]} \geq O$ and $R_{[k]} \geq O$, we must have $T_{[k]} = R_{[k]} = O_k$. We can use a similar argument for $B=O$ and show that $C = D = O$. Thus, $T$ and $R$ take the forms in Eq.~(\ref{eq:block-form}).

From Eq.~(\ref{eq:block-form}), we can deduce that the eigenvalues of $A$, $R$, and $T$ take the form:
\begin{equation}
\tilde{\bm{\alpha}} = (\bm{1}_k, \tilde{\bm{\alpha}}_{[n-k]}), \quad
\tilde{\bm{\rho}} = (\tilde{\bm{\rho}}_{[n-k]}, \bm{0}_k), \quad
\tilde{\bm{\tau}} = (\tilde{\bm{\tau}}_{[n-k]}, \bm{0}_k),
\label{eq:eigenvalue-form}
\end{equation}
where $\tilde{\bm{\alpha}}_{[n-k]}$, $\tilde{\bm{\rho}}_{[n-k]}$, and $\tilde{\bm{\tau}}_{[n-k]}$ are the eigenvalues of $A_{[n-k]}$, $R_{[n-k]}$, and $T_{[n-k]}$ respectively.

From the block-diagonal form in Eq.~(\ref{eq:block-form}),  we obtain the passivity condition:
\begin{equation}
A \geq O \iff A_{[n-k]} \geq O, \quad T \geq O \iff T_{[n-k]} \geq O, \quad R \geq O \iff R_{[n-k]} \geq O,
\label{eq:ATR-positivity}
\end{equation}
and the energy conservation condition:
\begin{equation}
A + R + T = I \iff A_{[n-k]} + R_{[n-k]} + T_{[n-k]} = I_{n-k}.
\label{eq:energy-conservation-reduced}
\end{equation}
Equations (\ref{eq:ATR-positivity}) and (\ref{eq:energy-conservation-reduced}) show that $A_{[n-k]}$, $R_{[n-k]}$, and $T_{[n-k]}$ satisfy the same conditions as $A$, $R$, and $T$ for an $(n-k)$-port system. Therefore, the set of all attainable tuples $(\tilde{\bm{\tau}}_{[n-k]}, \tilde{\bm{\rho}}_{[n-k]}, \tilde{\bm{\alpha}}_{[n-k]})$ is identical to $\Omega_{n-k}$. This completes the proof of Eq.~(\ref{eq:Omega_n_k-CPA}).
\end{proof}

\bibliography{main}

\end{document}